\documentclass[10pt, conference,usletter]{IEEEtran}
\usepackage{cite}
\ifCLASSINFOpdf
\else
\fi

\usepackage{graphicx}
\usepackage{epstopdf}
\usepackage{epsfig}
\usepackage[cmex10]{amsmath}
\usepackage{amssymb}
\usepackage{amsthm}
\usepackage{balance}

\usepackage[font=footnotesize, belowskip=2pt, aboveskip=0pt]{caption}
\newtheorem{thm}{Theorem}

\usepackage[square, comma, numbers]{natbib}
\usepackage{algorithm}
\usepackage{algorithmic}
\usepackage{subfigure}

\usepackage{array}
\ifCLASSOPTIONcompsoc
  \usepackage[caption=false,font=normalsize,labelfont=sf,textfont=sf]{subfig}
\else
  \usepackage[caption=false,font=footnotesize]{subfig}
\fi

\usepackage[draft,footnote,nomargin]{fixme}
\usepackage{url}
\begin{document}
\title{Backward-Shifted Strategies Based on SVC for HTTP Adaptive Video Streaming}
\author{\IEEEauthorblockN{Zakaria Ye\IEEEauthorrefmark{1},
Rachid El-Azouzi\IEEEauthorrefmark{1},
Tania Jimenez\IEEEauthorrefmark{1},
Eitan Altman\IEEEauthorrefmark{2} and 
Stefan Valentin\IEEEauthorrefmark{3} 
} 
\IEEEauthorblockA{\IEEEauthorrefmark{1}University of Avignon, Avignon, France, 
\IEEEauthorrefmark{2}INRIA, Paris, France, 
\IEEEauthorrefmark{3}Huawei Technologies, France \\
Email: \{zakaria.ye, rachid.elazouzi, tania.jimenez\}@univ-avignon.fr} eitan.altman@inria.fr, stefan.valentin@huawei.com}
\maketitle
\begin{abstract}
Although HTTP-based video streaming can easily penetrate firewalls and profit from Web caches, the underlying TCP may introduce large delays in case of a sudden capacity loss. To avoid an interruption of the video stream in such cases we propose the Backward-Shifted Coding (BSC). Based on Scalable Video Coding (SVC), BSC adds a time-shifted layer of redundancy to the video stream such that future frames are downloaded at any instant. This pre-fetched content maintains a fluent video stream even under highly variant network conditions and leads to high Quality of Experience (QoE). We characterize this QoE gain by analyzing initial buffering time, re-buffering time and content resolution using the Ballot theorem. The probability generating functions of the playback interruption and of the initial buffering latency are provided in closed form. We further compute the quasi-stationary distribution of the video quality, in order to compute the average quality, as well as temporal variability in video quality. Employing these analytic results to optimize QoE shows interesting trade-offs and video streaming at outstanding fluency.
\end{abstract}
\begin{IEEEkeywords} Scalable Video Coding, Quality of Experience, Queuing Theory, Ballot theorem, QoE Optimization\end{IEEEkeywords}
\IEEEpeerreviewmaketitle

\section{Introduction}
Recent studies show that video streaming already generated 45\% of all mobile data traffic in 2014 \cite{Cisco14}. By 2019, this fraction will likely increase to 62\%, while an 11-fold increase is predicted for the overall mobile data traffic \cite{Cisco14}. Despite the recent advances in increasing wireless capacity, this massive traffic load will drive mobile networks further into saturation and will turn user satisfaction into an enormous challenge. Consequently, more and more content providers and network operators will focus on Quality of Experience (QoE) per unit cost as primary metric for operational efficiency \cite{machine}.

The dominating factors of QoE are widely studied \citep{TKSK09, AMEI10}. Recent works developed approaches for understanding user engagements metrics \cite{DOB11, Sigcomm13}. The direct relation between time spent in rebuffering and user engagement was shown in \cite{Sigcomm13}. In \cite{DOB11}, the buffering ratio, rate of buffering, start up delay, rendering quality and average bit rate were demonstrated to show a dominating effect on QoE. Guidance to operators for improving user engagement in real time using only network-side measurements is provided in \cite{Sigmetric13}.  Authors of \cite{Yim11} found that temporal quality variation is worse than keeping a constant quality that is lower on the average. 

In general, understanding the QoE of mobile video is a complex task due to the many relationships between metrics and end-user's perceived video quality, metric-to-metric dependencies and confounding factors \cite{Sigmetric13}. At the same time, highly dynamic load and channel states lead to fluctuating capacity not only in mobile networks. 

To efficiently trade-off fluency and visual quality, more and more content providers deploy HTTP Adaptive Streaming (HAS) solutions, which are standardized as MPEG Dynamic Adaptive Streaming over HTTP (DASH) \cite{dash}. With DASH, each video file is divided into multiple small segments and each segment is encoded into multiple quality levels. Based on the available capacity, the client adaptively chooses the quality level of the segment such that visual quality is maximized at a low risk for an empty playback buffer.

The video segments can be created by encoding video content with various compression algorithms, where those following H.264/AVC (Advanced Video Coding) and H.264/SVC (Scalable video coding) \citep{codec, codec2} are widely employed. Although each AVC encoding run generates only segments of a single bitrate, segments for various bitrates can be created in multiple runs and chosen adaptively with HAS. This leads to a multiple bitrate video stream even with a non-scalable technique as AVC. On the other hand, SVC directly supports multiple bitrates within a single segment by multi-layer coding. With this technique, the video stream is encoded in one base layer and one or more enhancement layers. The base layer is always requested and, at sufficient capacity, one or more enhancement layers are additionally requested.

\subsection{Related Literature}
One benefit of HAS video streaming is that HTTP traffic can easily penetrate firewalls and profit from Web-infrastructure such as proxies and Content Delivery Networks (CDNs). The drawback, however, is that the underlying TCP protocol may introduce substantial delays to cope with packet errors and contention. For the video, this means that pixels errors and frame drops can essentially be ignored while resolution, initial buffering latency, starvation duration, and rate of buffering become the dominating factors for QoE. Such latencies are the result of an empty playback buffer as a consequence of choosing a higher quality than the supported bitrate \cite{Huang14}. To avoid such erroneous adaptation, current HAS policies are based on the measured segment fetch time that allows an instantaneous adaptation \cite{Bouaziz}. Authors in \cite{Zhou14} developed a rate adaptive method to enhance DASH performance over multiple content distribution servers. A Fuzzy-based controller has been proposed to dynamically adapt the video bitrate based on both the estimated  throughput and the size of the playback buffer \cite{sobhani15}. Other approaches have been explored by jointly considering the characteristics of the media content and the available wireless resources in the operator network \cite{Essaili, Lotte}.    

This paper proposes a complementary solution to DASH, named Backward-Shifted Coding (BSC). This solution makes HAS more robust to rapid fluctuations of the network capacity and provides more flexibility in increasing the quality of video without playback interruption. Furthermore, we develop an exact approach to obtain the distributions of the number of the playback interruptions and of the initial buffering latency. This analysis is closed to that of M/M/1 queue model in \cite{MULTI} and allows us to obtain an explicit formulation for the QoE metrics. While the bounds on the playback interruption probability were obtained in \cite{Paran} for an M/D/1 queue, our paper provides a new analysis for BSC which also obtains the average quality and temporal variability in video quality by using the quasi-stationary regime.  

\subsection{Main Contributions}
This paper provides important insight in optimization of HTTP adaptive video streaming and proposes a new Backward-Shifted Coding (BSC) scheme. The main idea of BSC is to add a layer of time-shifted redundancy to the video stream such that future frames are downloaded at any instant. In case of a sudden capacity drop, these pre-fetched frames can be played back and maintain a smooth video stream at sufficient quality. BSC can be implemented using the standard codecs and allows us to analytically obtain the key QoE factors by using the Ballot theorem.  Using these factors as inputs we finally propose a QoE optimization function.

We can summarize the main contributions of this paper as follows:
\begin{itemize}
\item[1.] We develop a novel coding scheme to improve the user QoE in HTTP adaptive streaming.
\item[2.] We present explicit form expressions for the QoE metrics. 
\item[3.] We propose an optimization scheme that take into account not only the waiting time, but also the mean video quality.
\item[4.] We show that our scheme can render a better QoE than existing bitrate adaptation algorithms used in DASH.
\item[5.] We show that our scheme can greatly reduce the probability of video playback interruption using several frames arrivals processes.
\end{itemize}

\subsection{Paper Organization}
The remainder of this paper is structured as follows: Section \ref{model} describes the system model. Section \ref{math} presents the analytical model for computing the QoE metrics using the Ballot theorem. Section \ref{optimization} presents the optimization issue for the QoE metrics while Section \ref{simulation} verifies the theoretical results and shows some numerical examples. Section \ref{conclusion} concludes the paper.

\section{System Design}
\label{model}
In this section we describe how our scheme Backward-Shifted Coding (BSC) can be used with any video codec that follows the H.264/SVC standard \citep{codec2}. Then we describe the integration into HAS based on the MPEG-DASH standard \cite{dash}.

\subsection{Mapping from BSC Scheme to Coding Scheme}
BSC is entirely client driven and it is independent of the video compression standard. The main idea of this scheme is to shift the base layer frames (low quality) and the enhancement layer frames (optimal quality), so that, when an interruption of playback buffer occurs, the base layer frames can still be played. To each frame $n$, we add its base layer in some subsequent frame $n-\phi+1$. If the starvation happened at frame $n$, the playback retrieves the base layer frame from frame $n-\phi+1$. In particular, we exploit temporal redundancy between subsequent frames in order to avoid the interruption of the playback buffer. Our scheme is inspired from Forward Error Correction (FEC) where the encoder adds redundancy to a message.  However, using the SVC codec, the BSC scheme does not generate any  overhead or redundant frames compared to the FEC scheme. Indeed,  with an SVC codec, the video bitstream contains a base layer and number of enhancement layers. The enhancement layers are added to the base layer to further improve the quality of video by increasing video frame-rate, temporal and quality scalability and spatial resolution. In our scheme, base layer and enhancements layers are temporally shifted as shown in Figure \ref{SVC}. We assume $\phi$ to be the offset between the basic layer frame and its enhancement layers. Frame $n$ with $1\leq n\leq \phi-1$ contains two blocks: complete block $n$ and the base layer of block $n+\phi-1$.  Frame $n$ with $n > \phi-1$ contains two blocks: the enhancement layers of block $n$ and the base layer of  block $n+\phi-1$.
\begin{figure}[t]
\begin{center}
\includegraphics[scale=0.5]{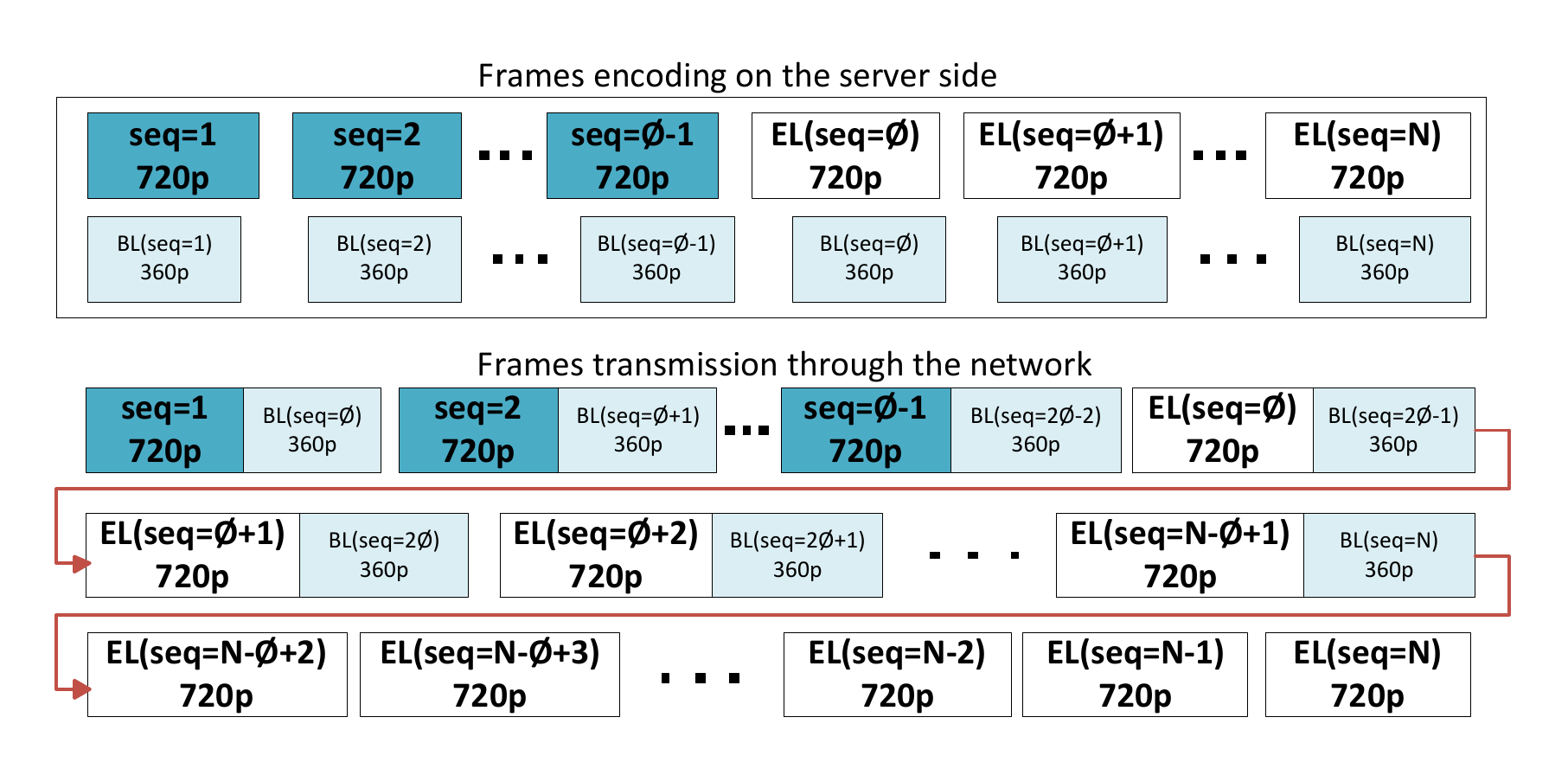}
\caption{Using SVC in Backward-Shifted Coding}
\label{SVC}
\end{center}
\end{figure}
At the user side, incoming bits are reassembled into video frames by the decoder. Starvation under BSC can happen at block $n$ if the base layer block $n$ is missing and the quality switching occurs when the enhancement layer is missing and the player finds only the base layer block $n$. 
 
\subsection{Mapping from BSC Scheme to DASH}
In DASH systems, each video consists of  multiple small segments at the media server and each segment is encoded at multiple discrete bitrates.  Hence, the BSC  scheme can be used at client side with  DASH.   In fact, when using the BSC scheme with DASH, the  video player decides the most suitable quality level of the base layer and the number of enhancement layers.  For example, given three video resolutions 720p, 480p and 360p,  the video  player can  decide  the most suitable configuration between the base layer and enhancement layer such as   360p (BL) and 480p (1 EL) or 360p (BL) and 720p (2 ELs).   This gives more flexibility to the video player to detect the highest quality level the current network conditions can support.

\section{Mathematical Performance Evaluation}
\label{math}
In order to evaluate the efficiency  of BSC using the SVC codec, in this section, we  develop a novel performance evaluation for QoE based on the Ballot theorem \citep{Feller_book}. This analysis captures QoE factors such as fraction of time spent rebuffering, content resolution and initial buffering latency. The analysis of starvation is closely related to analyzing the busy period in transient queues but differs in two aspects: First, our work aims to find the probability generating function of starvation events and not the queue size. Second, we do not assume a stationary arrival process.

\subsection{Starvation Analysis}
In this section, we call optimal frame the enhancement layers. The non-optimal frame corresponds to the base layer. We assume $N$ to be the media file size. When the streaming packets traverse the network, their arrivals to the media player are not deterministic due to the dynamic of the available bandwidth. The packets are reassembled by the decoder to render the video frames. We assume a Poisson distribution to describe the frames arrivals. After the streaming frames are received, they are first stored in the playout buffer. The interval between the service of two frames is assumed to be exponentially distributed so that we can model the receiver buffer as an M/M/1 queue.  The exponential distributed assumption is not the most realistic way to describe frame arrivals, but it reveals the essential features of the system, and it is the first step for more general arrival processes.  In Section \ref{simulation} we evaluate the performance of the Backward-Shifted Coding system by simulation, using different types of packet arrivals process such as the logistic process and the on-off process. The logistic process fits the video streaming traffic on the Long Term Evolution (LTE) networks according to \citep{streaming}

The maximum buffer size is assumed to be large enough to exclude buffer overflows. $\phi$ is the offset between the optimal frame and its corresponding non-optimal frame (Fig. \ref{fec}). A starvation happens when the playout buffer is empty. \\
We denote by $\lambda$ the Poisson arrival rate of the frames, and by $\mu$ the Poisson service rate. We define $\rho = \lambda/\mu$ to be the traffic load.\\
In a non-empty M/M/1 queue with everlasting arrivals, the rate at which either an arrival or a departure occurs is given by $\lambda+\mu$. This event corresponds to an arrival with probability $p$, or is otherwise to an end of service with probability $q$, where $$ p= \frac{\lambda}{\lambda+\mu } = \frac{\rho}{1+\rho}; \quad \quad q = \frac{\mu }{\lambda+\mu } = \frac{1}{1+\rho} $$ The buffer is initially empty. We let $ T_x $ be the initial buffering delay, in which $ x $ frames are accumulated in the buffer.

\subsubsection{\textbf{Probability of Starvation}}
We present a frame level model to investigate the starvation probability with the BSC. Our analysis of the probability of starvation is built on the Ballot theorem\footnote{Ballot theorem: In a ballot, candidate $ A $ scores $ N_A $ votes and candidate $ B $ scores $ N_B $ votes, where $ N_A > N_B $. Assume that while counting, all the ordering (i.e. all sequences of $ A $'s and $ B $'s) are equally alike, the probability that throughout the counting, $ A $ is always ahead in the count of votes is $ \frac{N_A - N_B }{N_A + N_B} $.}.
\vspace{0.2cm}\\ 
Since we set the value of $\phi$ at the beginning of the video session, the starvation can happen before the arrival of frame $\phi$ if $\phi > x$. So we have to investigate the probability of starvation by distinguishing the two cases: $\phi \leq x$ and $\phi > x$.  Let $P_{s}^{<}(N,\phi,x)$ and $P_{s}^{>}(N,\phi,x)$ denote, respectively, the probability of starvation for  $\phi \leq x$ and $\phi > x$. 
\begin{figure*}[htb!]
\begin{center}
\includegraphics[scale=0.6]{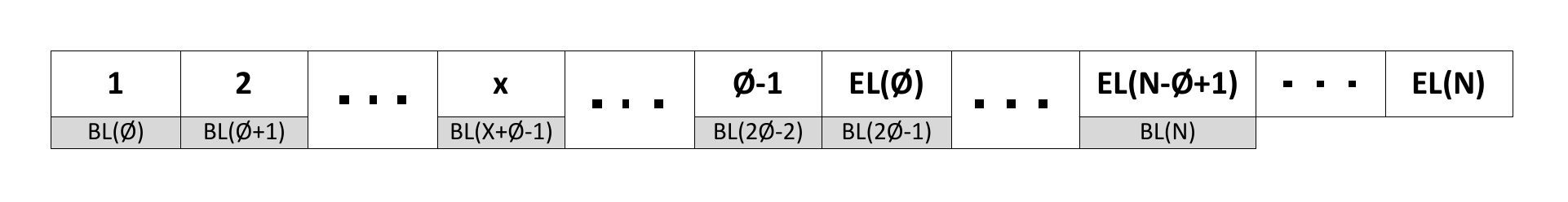}
\caption{\label{fec}The BSC Coding for the offset $\phi > x$. The starvation can happen before the arrival of optimal frame $\phi$ since $x < \phi$.}
\end{center}
\end{figure*}
 For $\phi \leq x$, the media player starts to work  when the number of optimal frames in the buffer reaches $x$ optimal frames which corresponds to $x+\phi-1$ non-optimal frames stored in the buffer.    Thus $P_{s}^{<}(N,\phi,x)$ is given by  \citep{MULTI}, 
\begin{multline}
\label{starvation1}
P_{s}^{<}(N,\phi,x) = \sum_{k=x+\phi-1}^{N-1}{\frac{x+\phi-1}{2k-x-\phi+1}\binom{2k-x-\phi+1}{k-x-\phi+1}}\cdot \\
p^{k-x-\phi+1}q^k.
\end{multline}
Let  $P_{x,\phi}^k$ and $P_{x}^k$ denote, respectively, the probability that the starvation happens exactly after the departure  of frame $k$ using BSC and without BSC.  These probabilities are given by 
\begin{multline}
\label{avecfec}
P_{x,\phi}^k= \frac{x+\phi-1}{2k-x-\phi+1}\binom{2k-x-\phi+1}{k-x-\phi+1} \cdot
p^{k-x-\phi+1}q^k,\\
P_{x}^k=\frac{x}{2k-x}\binom{2k-x}{k-x}\cdot p^{k-x}\cdot q^k. \hspace{2.88cm}
\end{multline} 
\begin{thm}
\label{thrm2}
For the offset $\phi > x$, the probability of starvation is given by: 
\begin{multline}
\label{starvation2}
P_{s}^{>}(N,\phi,x) = P_{s1}+(1-P_{s1})\cdot P_{s2}
\end{multline}
where
\begin{equation}
\label{avant}
P_{s1} = \sum_{k=x}^{\phi-2}{\frac{x}{2k-x}\binom{2k-x}{k-x}\cdot p^{k-x}\cdot q^k},
\end{equation}
and
\begin{multline}
\label{apres}
P_{s2}= \sum_{k=2\phi-2}^{N-1}{\frac{x+\phi-1}{2k-x-\phi+1}\binom{2k-x-\phi+1}{k-x-\phi+1}}\cdot\\
p^{k-x-\phi+1}q^k.
\end{multline}
\end{thm}

\begin{proof}
For the case $\phi > x$, the starvation could happen before the arrival of frame $\phi$. We define $E_{<\phi}$ and $E_{>\phi}$ to be the event that the starvation happens for the first time before $\phi$ and after $\phi$, respectively. The event of starvation is $ E_{<\phi}\cup (E_{>\phi} \cap \bar{E}_{<\phi}) $; where $\bar{E}_{<\phi}$ is the complementary of $E_{<\phi}$ and is the event that no starvation happens before the arrival of frame $\phi$. We have $ P(E_{<\phi}\cup (E_{>\phi} \cap \bar{E}_{<\phi}))=P(E_{<\phi})+P(E_{>\phi} \cap \bar{E}_{<\phi}) $.
For the event $E_{<\phi}$, since we cannot use the non-optimal frames of the BSC coding because the starvation happens before $\phi$, the probability of starvation is $P_{s1}$. We exclude in this sum $\phi-1$ since the starvation cannot happen after the service of frame $\phi-1$ because the frame $\phi$ (non-optimal) is already stored in the buffer. 
Now we compute the probability of the second term $P(E_{>\phi} \cap \bar{E}_{<\phi})$. We have
$ P(E_{>\phi} \cap \bar{E}_{<\phi})= P(E_{>\phi} / \bar{E}_{<\phi})P(\bar{E}_{<\phi}) $.
Since $\bar{E}_{<\phi}$ is the complementary of $E_{<\phi}$, $P(\bar{E}_{<\phi})$ is $1-P_{s1}$. $E_{>\phi} / \bar{E}_{<\phi}$ is the event that a starvation happens for the first time after the arrival of frame $\phi-1$, given that the starvation does not happen before. Then, assuming that the starvation does not happen before $\phi-1$, frame $x+\phi-1$ will be played after $\phi-1$ as shown in Fig. \ref{thmfig}. At this instant, the non-optimal frame $2\phi-2$ is already in the buffer, so the starvation cannot happen before the service of $2\phi-2$. We use the Ballot theorem to compute the probability of the event $E_{>\phi} / \bar{E}_{<\phi}$ based on the non-optimal frames. The most important trick is the origin of the Ballot theorem, i.e., where to start the process of counting the frames arrivals and departures. The inappropriate method is to start the counting process just after the arrival of the frame $\phi-1$. At that moment, we do not know the number of departures that occur before. So we start the counting process when we have $x$ optimal frames in the buffer, that correspond to the last non-optimal frame $x+\phi-1$. We define $A_k$ to be an event that the buffer becomes empty for the first time when the service of frame $k$ is finished (Fig. \ref{thmfig}).
\begin{figure}[b]
\begin{center}
\includegraphics[scale=0.7]{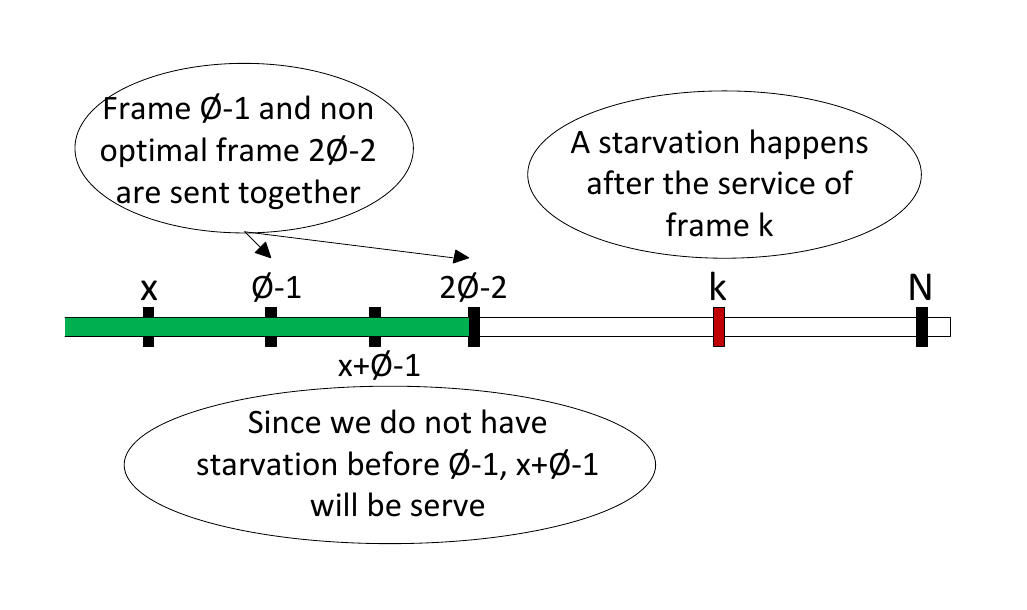}
\caption{\label{thmfig}If the starvation does not happen before $\phi-1$, then there will be no starvation untill the service of frame $2\phi-2$.}
\end{center}
\end{figure}
All the events $A_k, k = 1, ...N$, are mutually exclusive. The event of starvation is the union $ \cup_{k=2\phi-2}^{N-1}{A_k} $. We exclude in this union $A_k$ for $k \in [1, 2\phi-3]$ because we cannot have a starvation before the arrival of optimal frame $\phi-1$. That corresponds to the non-optimal frame $2\phi-2$. This union of events excludes $E_N$ because the empty buffer after the service of $N$ packets is not a starvation. When the buffer is empty at the end of the service of the $k^{th}$ packet, the number of arrivals is $k-x-\phi+1$ after the pre-fetching process. The probability of having $k-x-\phi+1$ arrivals and $k$ departures is computed from the binomial distribution $\binom{2k-x-\phi+1}{k-x-\phi+1}\cdot p^{k-x-\phi+1}\cdot q^k$. For the necessary and sufficient condition of the event $A_k$, we apply the Ballot theorem. If we count the number of arrivals and departures when the playback starts, the number of departures is always greater than the number of arrivals. Otherwise, the empty buffer already happens before the $ k^{th} $ frame is served. According to the Ballot theorem, the probability of event $A_k$ is computed by $\frac{x+\phi-1}{2k-x-\phi+1}\binom{2k-x-\phi+1}{k-x-\phi+1}\cdot p^{k-x-\phi+1}\cdot q^k$. Therefore, the probability of the event $E_{>\phi} / \bar{E}_{<\phi}$,  is the probability of the union $ \cup_{k=x+\phi-1}^{N-1}{A_k} $, given by (\ref{apres}).
\end{proof}

\subsubsection{\textbf{Probability Generating Function of Starvation Events}}
\begin{figure}[b]
\begin{center}
\includegraphics[scale=0.6]{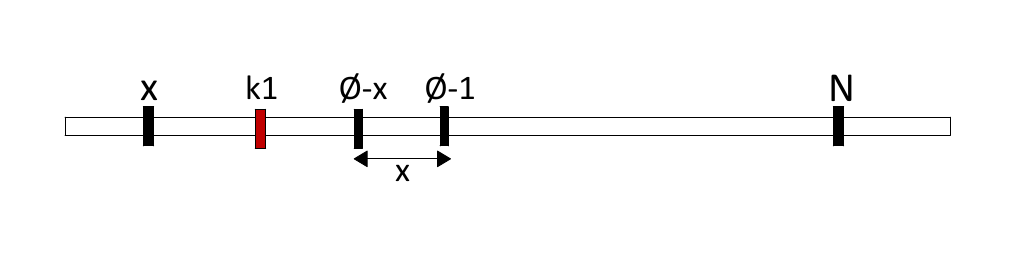}
\caption{\label{jstarv}A starvation happens at $k_1 < \phi-x$.}
\end{center}
\end{figure}
\begin{figure}[ht]
\begin{center}
\includegraphics[scale=0.6]{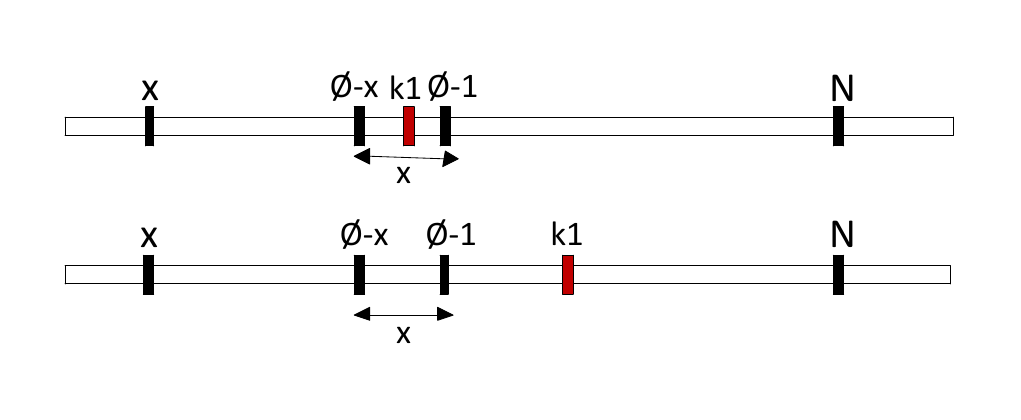}
\caption{\label{jstarv2}A starvation happens at $k_1 > \phi-x$.}
\end{center}
\end{figure}
In the BSC scheme, the starvation may happen for more than once during the file transfer. We are interested in the probability distribution of starvation, given the finite file size $N$. When $\phi \leq x$, we cannot have a starvation before the service of $\phi-1$. In this case, the probability generating function is similar to that of \cite{MULTI} in replacing $x$ by $x+\phi-1$. \\
Now, we show how the probability generating function of starvation events can be derived using the Ballot theorem for the case $\phi > x$. We define a path as a complete sequence of frame arrivals and departures \cite{MULTI}. The probability of a path depends on the number of starvations. We consider a path with $ j $ starvations. To carry out the analysis, we start from the event that the first starvation takes place. We denote by $ k_l $ the $ l^{th} $ departure of a frame that sees an empty queue. We notice that the path can be decomposed into the following mutually exclusive events:
\begin{itemize}
\item[.] Event $ \mathcal{F} (k_1) $: the buffer becoming empty for the first time in the entire path.
\item[.] Event $ \mathcal{M}_l (k_l,k_{l+1}) $: the empty buffer after the service of frame $ k_{l+1} $ given that the previous empty buffer happens at the departure of frame $ k_l $.
\item[.] Event $ \mathcal{L}_j (k_j) $: the last empty buffer observed after the departure of the last frame $ k_j $. 
\end{itemize}
We let $ P_{\mathcal{F}(k_1)} $, $ P_{\mathcal{M}_{l}(k_{l},k_{l+1})} $ and $ P_{\mathcal{L}_j(k_j)} $ be the probabilities of events $ \mathcal{F}(k_1) $, $ \mathcal{M}_{l}(k_{l},k_{l+1}) $ and $ \mathcal{L}_j(k_j)$, respectively, and will analyze the probabilities of these events step by step. We first compute the probability of having only one starvation. This probability concerns the two events: $\mathcal{F}(k_1)$ and $\mathcal{L}_1(k_1)$, i.e., the event that the buffer becomes empty for the first time after the service of the frame $k_1$ and the event that we do not observe an empty buffer after $k_1$ until the end of the video file (Fig. \ref{jstarv} and \ref{jstarv2}).

The starvation can happen at $k_1$ before we use the BSC redundant frames, i.e., before the arrival of frame $\phi-1$. In this case, the probability of starvation is given by $P_{x}^{k_1}$ of (\ref{avecfec}). If the starvation happens after the arrival of frame $\phi-1$, then it necessarily happens after the service of frame $2\phi-3$ since we cannot have a starvation between $\phi-2$ and $2\phi-3$. Then the probability of starvation is given by Theorem \ref{thrm2}. So the probability distribution of event $ \mathcal{F}(k_1) $ is expressed as
\begin{equation}
P_{\mathcal{F}(k_1)} :=
\left \{
\begin{array}{c l l}
0, \quad \quad if ~~ k_1 < x ~~ or ~~ k_1 = N; \\
P_{x}^{k_1} \quad if ~~k_1 \in [x,...,\phi-2]; \\
0, \quad if ~~ k_1 \in [\phi-1,...,2\phi-3]; \\
(1-P_{s1})P_{x,\phi}^{k_1}, \\
if ~~k_1 \in [2\phi-2,...,N-1].
\end{array}
\right.
\end{equation}
$P_{x}^{k_1}$, $P_{x,\phi}^{k_1}$ and $P_{s1}$ are given by (\ref{avecfec}) and (\ref{avant}).
Given that the only starvation happens at $k_1$, what is the probability that no starvation happens until the end of the video? That is the probability of the event $\mathcal{L}_1(k_1)$. We take the complement of starvation probability as the probability of no starvation for the file size $N-k_1$. We denote $x_{\phi} = x+\phi-1$ to simplify the expressions. We distinguish two cases (Fig. \ref{jstarv} and \ref{jstarv2}). The first case is that the starvation happens at $k_1$, and we still can have a starvation before the arrival of frame $\phi-1$ (Fig. \ref{jstarv}). Then, we use the probability of starvation of the theorem \ref{thrm2}, $P_{s}^{>}(N-k_1,\phi,x)$ to compute the probability of having no starvation until the end of the video file. For the remaining case (Fig. \ref{jstarv2}), a starvation cannot happen before we use the non-optimal frames. Then, we use the probability of starvation, $P_{s}^{<}(N-k_1,\phi,x)$. Finally, the probability distribution of event $ \mathcal{L}_1(k_1) $ is expressed by
\begin{equation}
P_{\mathcal{L}_1(k_1)} =
\left \{
\begin{array}{l l l}
0, \quad \quad if ~ k_1 < x ~ or ~ k_1 = N; \\
1-P_{s}^{>}(N-k_1,\phi,x), \quad \quad if ~ x \leq k_1 < \phi-x; \\
1-P_{s}^{<}(N-k_1,\phi,x), ~~ if \\
\small{\phi-x \leq k_1 < \phi-1 ~or~ 2\phi-2 \leq k_1 < N-x_{\phi};} \\
0, \quad if \quad \phi-1 \leq k_1 < 2\phi-3; \\
1, \quad \quad if \quad N-x_{\phi} \leq k_1 < N.
\end{array}
\right.
\end{equation}
We denote by $ P_s(j) $ the probability of having $ j $ starvations. For the case with one starvation, $ P_s(1) $ is solved by 
\begin{equation}
P_s(1) = \sum_{i=1}^{N}{P_{\mathcal{F}(i)}P_{\mathcal{L}_1(i)}} = \textbf{P}_{\mathcal{F}}\cdot \textbf{P}_{\mathcal{L}_1}^{T}
\end{equation}
where $ ^T $ denotes the transpose. Here, $ \textbf{P}_{\mathcal{F}} $ is the row vector of $ P_{\mathcal{F}(i)} $, and $ \textbf{P}_{\mathcal{L}_1} $ is the row vector of $ P_{\mathcal{L}_1(i)} $, for $ i = 1, 2, ..., N $.
Now we compute the probability of having more than one starvation. A path with $ j $ starvations is composed of a succession of events $$ \mathcal{F}(k_1), \mathcal{M}_1(k_1,k_2),...,\mathcal{M}_{j-1}(k_{j-1},k_{j}),\mathcal{L}_j(k_j). $$ We have to solve the probability distribution of the events $\mathcal{L}_j(k_j) $ and $\mathcal{M}_{l}(k_{l},k_{l+1})$. Suppose that there are $ j $ starvations after the service of frame $ k_j $. Then, only the lower bound of $k_j$ changes in the event $\mathcal{L}_j(k_j) $ from $\mathcal{L}_1(k_1)$. This lower bound corresponds to the extreme case where the $ j $ starvations take place consecutively. Let $e$ be the number of starvations that we can have before frame $\phi-1$ in the extreme case. So $e=\lfloor \frac{\phi-2}{x} \rfloor$. We distinguish two cases where $e \geq j$ and $e < j$ (Fig. \ref{kjstarv}).
If $e \geq j$, then all the $j$ starvations will happen before the service of frame $\phi-1$ and the lower bound of $k_j$ is $jx$. Then we find the expression of $\mathcal{L}_j(k_j)$ in replacing the lower bound $x$ by $jx$ in the expression of $\mathcal{L}_1(k_1)$. If $e < j$, then the remaining starvations will happen after the service of frame $2\phi-2$ since we do not have starvation between $\phi-1$ and $2\phi-2$.
Hence, the probability distribution of event $ \mathcal{L}_j(k_j) $ is
\begin{equation}
P_{\mathcal{L}_j(k_j)} =
\left \{
\begin{array}{c l l}
0, \\
if ~ k_j < 2\phi-2+(j-e)x_{\phi} ~ or ~ k_j = N; \\
1-P_{s}^{<}(N-k_j,\phi,x), \\
if ~ 2\phi-2+(j-e)x_{\phi} \leq k_j < N-x_{\phi}; \\
1, \quad \quad if \quad N-x_{\phi} \leq k_j < N.
\end{array}
\right.
\end{equation}
We now compute the probability of the event $ \mathcal{M}_{l}(k_{l},k_{l+1}) $. After frame $ k_l $ is served, the $ l^{th} $ starvation is observed. We compare $k_l$ to $e$ to obtain the position of $k_l$. If $e < l$, then $k_l$ should not be less than $2\phi-2+(l-e)(x+\phi-1)$ in order to have $l$ starvations. Also, $k_{l+1}$ must satisfy $k_l + (x+\phi-1) \leq k_{l+1} < N-(j-l-1)(x+\phi-1)$ because of the pre-fetching after $k_l$ and the fact that we have $j$ starvations in total. In this case, $ P_{\mathcal{M}_{l}(k_{l},k_{l+1})} $ is expressed as $\frac{x+\phi-1}{2k_{l+1}-2k_l-x-\phi+1}\binom{2k_{l+1}-2k_l-x-\phi+1}{k_{l+1}-k_l-x-\phi+1}p^{k_{l+1}-k_l-x-\phi+1}q^{k_{l+1}-k_l}$. If $e \geq l$, then we have $lx \leq k_l \leq \phi-1$. For $k_l+x \leq k_{l+1}<\phi-1$, $ P_{\mathcal{M}_{l}(k_{l},k_{l+1})} $ is expressed as $\frac{x}{2k_{l+1}-2k_l-x}\binom{2k_{l+1}-2k_l-x}{k_{l+1}-k_l-x}p^{k_{l+1}-k_l-x}q^{k_{l+1}-k_l}$. Since the $(l+1)^{th}$ starvation cannot happen between $\phi-1$ and $2\phi-2$, for $2\phi-2 \leq k_{l+1} < N-(j-l-1)(x+\phi-1)$, the probability of the event $\mathcal{M}_{l}(k_{l},k_{l+1})$ is expressed as the probability for the case $e < l$.
We denote by $ \textbf{P}_{\mathcal{M}_{l}} $ the matrix of $ P_{\mathcal{M}_{l}(k_{l},k_{l+1})} $ for $ k_l, k_{l+1} \in [1, N] $. The probability of having $ j(j \geq 2)$ starvations is given by
\begin{figure}[t]
\begin{center}
\includegraphics[scale=0.7]{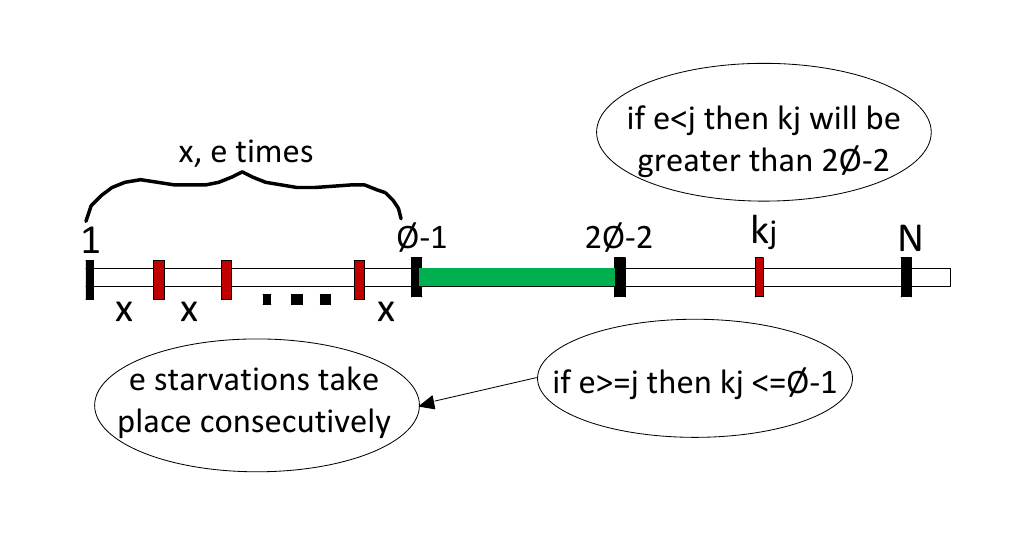}
\caption{\label{kjstarv}The lower bound of $k_j$ in case we have $j$ starvations.}
\end{center}
\end{figure}
\begin{multline}
P_s(j) = \sum_{k_1=1}^{N}\sum_{k_2=1}^{N}, \dots, \sum_{k_{j-1}=1}^{N}\sum_{k_j=1}^{N}{P_{\mathcal{F}(k_1)} \cdot P_{\mathcal{M}_{1}(k_{1},k_{2})}}, \dots, \\
P_{\mathcal{M}_{j-1}(k_{j-1},k_{j})}\cdot P_{\mathcal{L}_j(k_j)} = \textbf{P}_{\mathcal{F}}\cdot \bigg( \prod_{l=1}^{j-1}{\textbf{P}_{\mathcal{M}_l}} \bigg) \cdot \textbf{P}_{\mathcal{L}_j}^{T}.
\end{multline}
Then, we can write the probability generating function (p.g.f) $ G(z) $ by
\begin{equation}
G(z) = E(z^j) = \sum_{j=0}^{J}{P_s(j)\cdot z^j}.
\end{equation}

\subsection{Analysis of the Video Quality}
\label{mean_quality}
In this section, we analyze the average video quality of the BSC system. For this purpose we model the system as a continuous-time Markov birth-death process with an absorbing state which corresponds to the starvation event. Then, we compute the amount of time spent in each bitrate level. We  call low bitrate, the bitrate of the base layer frames (or non-optimal frames) and optimal bitrate, the bitrate of the combined base and enhancement layers frames.
The switching process is shown in Fig. \ref{markov2}. \\
Let $A$ and $D$ be the sequence number of the last optimal frame in the buffer and the sequence number of the last frame that was displayed at the screen respectively. The state $n$ of the Markov process is $A-D$. If $D \leq A$, the state $n$ is positive and there is exactly $A-D$ available optimal frames in the buffer. Otherwise, if $D > A$, the state $n$ is negative. In that case, there is no more available optimal frames and the number of available non optimal frames is $A-D+\phi$. For example if the process is in the state $-\phi+1$ (i.e., $A-D = -\phi+1$), it remains exactly 1 non optimal frame in the buffer as shown in Fig. \ref{markov2}. \\
\begin{figure}[t]
\begin{center}
\includegraphics[scale=0.5]{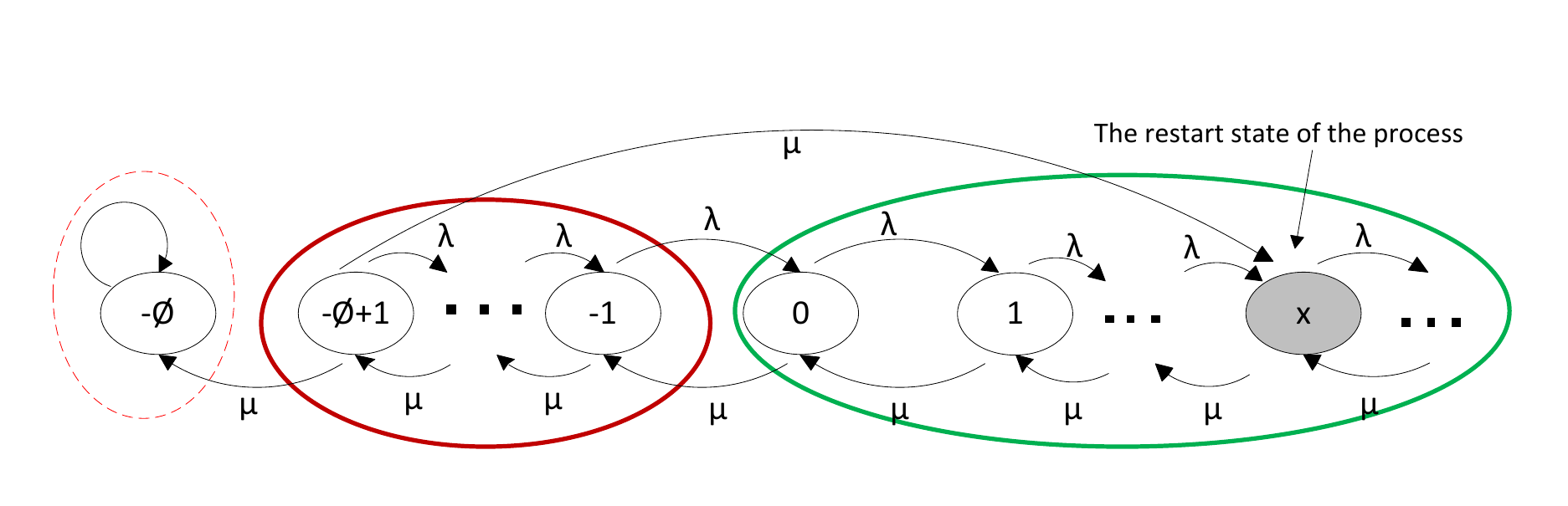}
\caption{\label{markov2} Markov model of the switching mechanism, denoting the difference between the number of optimal and non-optimal frames in the playout buffer.}
\end{center}
\end{figure}
The infinitesimal generator $Q$ of the process is a tridiagonal matrix  where   elements of the diagonal, upon the diagonal and under the diagonal are $-(\lambda+\mu)$, $\lambda$ and $\mu$ respectively, except the first element of the diagonal and upon the diagonal which is $0$ (due to the absorbing state).

First, we compute the time to absorption starting  from  initial state $k>-\phi$.  To do this, we substitute transition to the absorbing state $-\phi$ by transition  to the initial state $k$,  whenever the resulting Markov chain is ergodic and admits a stationary regime ${\bf q}=(q_{-\phi+1}, q_{-\phi+2},..)$. ${\bf q}$ is given by solving the balance equations of the resulting Markov chain
\begin{equation}
\label{balanceeq}
\left \{
\begin{array}{l c l}
\mu q_{-\phi+2}&=& (\lambda+\mu)q_{-\phi+1}\\
\lambda q_{j-1} + \mu q_{j+1}&=& (\lambda+\mu)q_{j}, j \neq x+1\\
\lambda q_{x-1} + \mu (q_{x+1}+q_{-\phi+1})&=& (\lambda+\mu)q_{x}\\
\sum_{j=-\phi+1}^{\infty}{q_j}=1
\end{array}
\right.
\end{equation}
\begin{eqnarray*}
q_j &=& q_{-\phi+1} \frac{1-\rho^{j+\phi}}{1-\rho},\;\; j=-\phi+1, -\phi+2,..,x\\
q_j&=&q_{-\phi+1}  \frac{1-\rho^{x+\phi}}{1-\rho}\rho^{j-x},\;\; j=x+1,x+2,...
\end{eqnarray*}
Thus, we have 
\begin{equation}
\label{qq}
\left \{
\begin{array}{l c l}
q_{-\phi+1}&=& \frac{1-\rho}{\phi+x}\\
q_j &=& \frac{(1-\rho^{j+\phi})}{(\phi+x)},\;\; j=-\phi+1,..,x\\
q_j&=&\frac{(1-\rho^{x+\phi})}{((\phi+x)}\rho^{j-x},\;\; j=x+1,x+2,...
\end{array}
\right.
\end{equation}

Let $E[S_i]$ be the expected time spent in state $i$ before the process reaches the absorbing state. From \cite{Lefevre}, we have
$$
E[S_i]= q_i  E[\tau_x] 
$$
where $E_{[\tau_{x}]}$ is the  time to absorption into state  $-\phi$ starting from the initial state $x$.   According to \citep{timeabsorption}, the mean time to absorption into state $-\phi$ from the initial state $x$ is given by 
 \begin{equation}
E[\tau_x]  = \frac{x+\phi}{\mu - \lambda}
\end{equation}
Thus the  expected time spent in state $i$ before the process reaches the absorbing state is given by 
\begin{equation}
E[S_i]= q_i  \frac{x+\phi}{\mu - \lambda}
\end{equation}
where $q_i$ is given by (\ref{qq}).    Thus the time spent in the low bitrate and in the optimal bitrate are  given, respectively,   by
\begin{equation}
T_{\mathcal{L}} = \sum_{i=-\phi+1}^{-1} E[S_i] \mbox{ and } T_{\mathcal{H}} = \sum_{i=0}^{\infty} E[S_i]
\label{qualityx}
\end{equation}

Let $b_{\mathcal{L}}$ and $b_{\mathcal{H}}$ be the bitrate of the low coding (base layer) and the bitrate of the optimal coding (base and enhancement layers) respectively.  Hence, the  average bitrate is 
\begin{equation}
b_{avg}= \frac{T_{\mathcal{L}}b_{\mathcal{L}}+b_{\mathcal{H}}T_{\mathcal{H}}}{T_{\mathcal{L}}+T_{\mathcal{H}}}
\end{equation}

Let us now compute the distribution of period of time during which the  video playback quality is optimal, named $B_H$.  This period corresponds to  the duration of time that the process starting from state 1, stays continuously away from state 0.  Using the analysis from the busy period in M/M/1, the  expected  and variance of  $B_H$  is given, respectively,  by \cite{timeabsorption}
$$
E[B_H]= \frac{1}{\mu-\lambda},\;\; \mbox{ and } V(B_H)=\frac{(1+\rho)}{\mu^2(1-\rho)^3}
$$

The analysis of the variation of the quality is omitted in this paper due to lack of space.

\subsection{The Initial Buffering Delay and Rebuffering Delay}
\label{sectionoffset}
The expected initial buffering delay is $\frac{x}{\lambda}$ and the expected rebuffering delay is $\frac{x+\phi-1}{\lambda}$.
Note that the BSC scheme increases the start-up delay compared to an equivalent single bitrate system because of the shift $\phi$. Indeed, we send the enhancement layer $\phi$ after the complete frame $\phi-1$. But in practice, this delay does not exceed a couple of seconds, then it does not have a huge impact on the overall quality of experience. However, $\phi$ impacts the rebuffering delay, i.e., the waiting time after a starvation event. The player starts when we accumulate $x$ optimal frames in the buffer. So a high value of $\phi$ increases the amount of rebuffering time.

\section{Evaluation of the QoE}
\label{optimization}
In this section, we show how the BSC can be used to improve the metrics of the quality of experience. The parameter $\phi$ affects the initial buffering delay and the starvation. We first confirm this with the model and show how one can choose $\phi$ to improve the QoE metrics. Then, we propose an evaluation function to improve the user experience during the video session.

\subsection{Choosing the Offset $\phi$}
To set the value of the offset $\phi$, we should know the probability of starvation of a simple system without BSC according to the file size $N$. This probability is given in \cite{MULTI} by 
\begin{equation}
\label{offset}
P_{starvation} = \sum_{k=x}^{N-1}{\frac{x}{2k-x}\binom{2k-x}{k-x}\cdot p^{k-x}\cdot q^k},
\end{equation}
It also corresponds to the probability of playing the non optimal frames for the first time in BSC system since a starvation in a system without BSC is equivalent to a switching in BSC system. The surrounded region in figure \ref{fecz} corresonds to the set of values of $N$ where the risk of the playback interruption is small ($\phi$ can be set up to 80 in the figure without risk of starvation).
\begin{figure}[t]
\begin{center}
\includegraphics[scale=0.38]{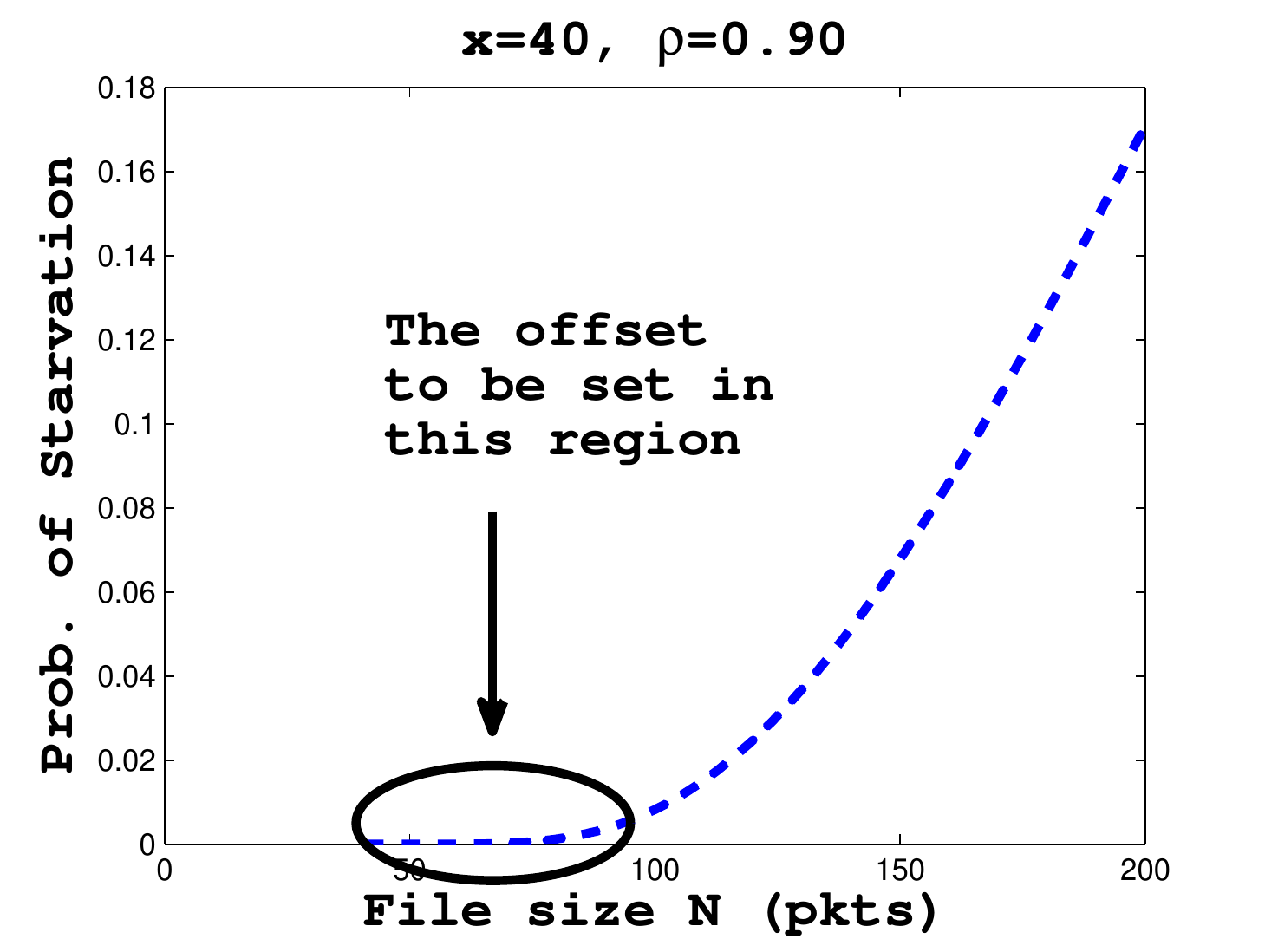}
\caption{\label{fecz}The value of the offset $\phi$.}
\end{center}
\end{figure}
Then, we can choose the offset $\phi$ to improve the QoE metrics: the rebuffering time, the starvation and the average bitrate. We have less risk of starvation for some values of $\phi$. Since after each starvation event, we wait for $x+\phi-1$ frames, small values of $\phi$ decreases the rebuffering time. Moreover, (\ref{qualityx}) shows that the time spent in low bitrate level increases with $\phi$. Then, $\phi$ must satisfy these QoE metrics requirements too. Parameter $\phi$ is  calculated  at the beginning of the video session based on the networks conditions and the user profile but its value can be updated after each starvation if the networks conditions changed. 

\subsection{Evaluation Scheme for the Global QoE}
To use our analysis for optimizing the QoE, we define an objective QoE cost function $C(x, N)$ for a user, which includes the expected initial buffering latency, the number of playback interruption during the video session and the average video quality. These metrics are weighted by three coefficients $\gamma_1$, $\gamma_2$ and $\gamma_3$, which allow us to balance the tradeoff among the QoE metrics according to user preferences. 

The weights $\gamma_1$ and $\gamma_2$ are preceded by a positive sign because the smaller the initial buffering latency and the number of starvation are, the better the QoE is. Implicit tradeoff exists between the two metrics and the initial buffering time is preferred to the starvation by around 90\% of users \citep{sdelay2}. Although users have a very low tolerance for playback interruptions \citep{youtubeqoe} they also only accept a start-up delay between 5 and 15 seconds, depending on the duration of the overall video \citep{sdelay2}.

The weight $\gamma_3$ is preceded by a negative sign because the higher the average quality, the better the QoE. 
According to this reasoning, we propose the cost function
\begin{equation}
C(x,N)= \gamma_1 \cdot E[T_x] + \gamma_2.P_s(j) - \gamma_3 \cdot \sum_{i=1}^{r}{w_i T_i}
\end{equation}
where $E[T_x]$ is the expected initial buffering delay, $P_s(j)$ is the number of starvations, $r$ is the number of available bitrates, $w_i$ is a weight associated to each bitrate and $T_r$ is the fraction of time spent in each bitrate level. We use this function to evaluate the BSC scheme performance on the QoE metrics, i.e., how our scheme can evaluate the objective QoE cost $C(x,N)$. The examples are shown in section \ref{QoEOPT}.

\section{Simulation and Numerical Examples}
\label{simulation}

\subsection{Simulation and validation}
\begin{figure}[t]
\begin{minipage}[t]{0.49\linewidth}
\centering\epsfig{figure=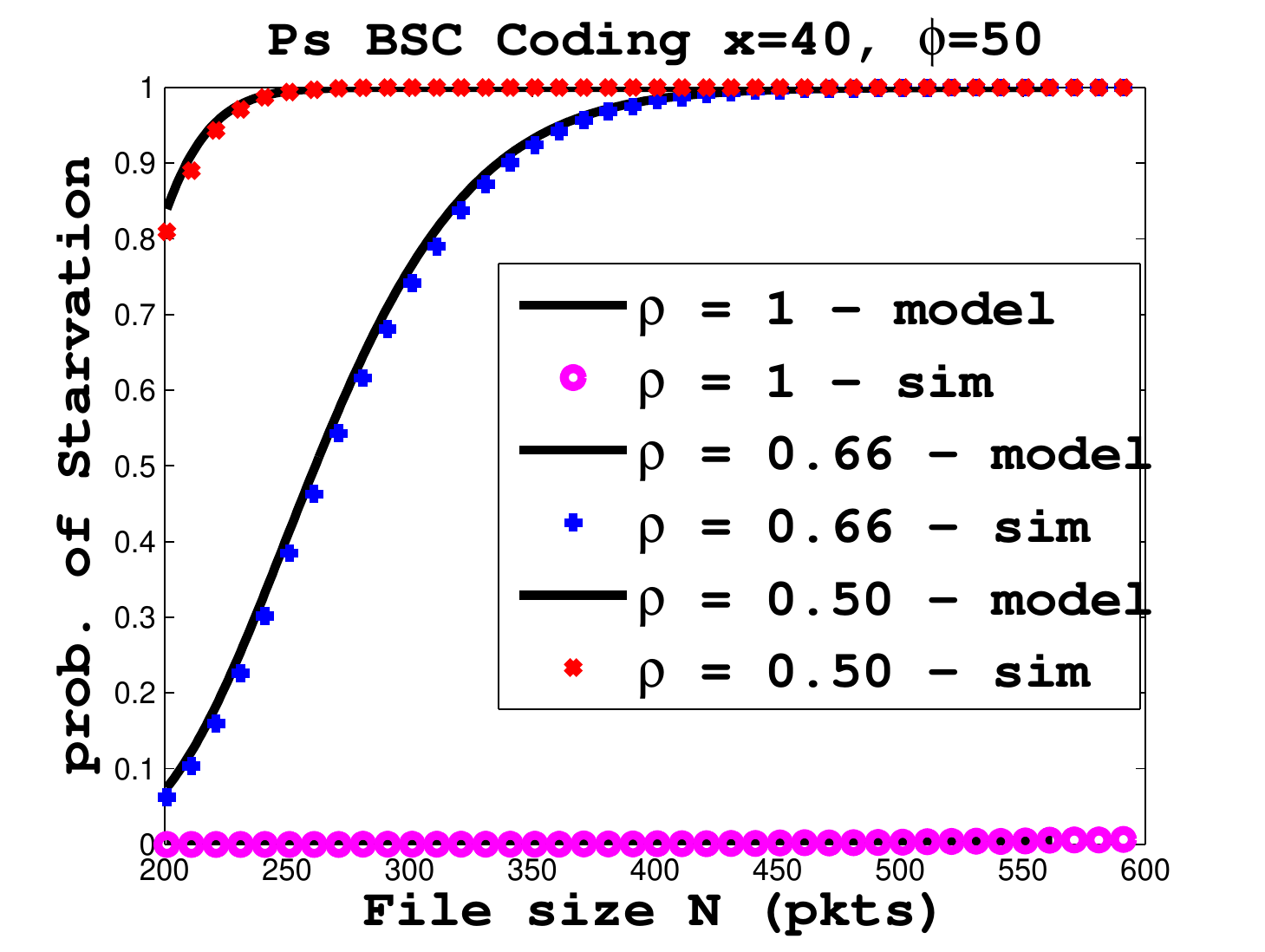, width=\linewidth}
\caption{The probability of starvation vs the file size N \label{starvfile}}
\end{minipage}
\begin{minipage}[t]{0.49\linewidth}
\centering\epsfig{figure=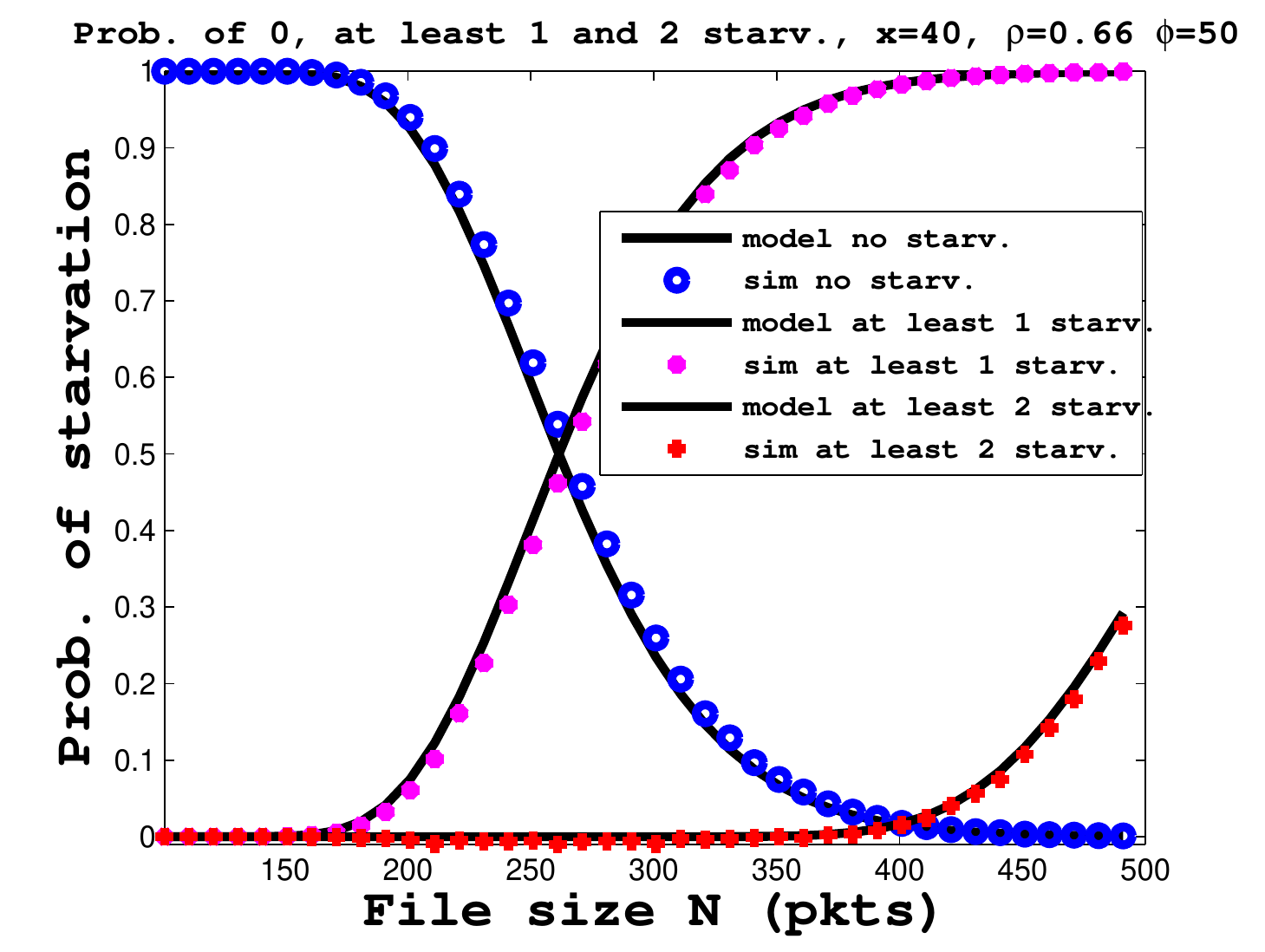, width=\linewidth}
\caption{The probability of no starvation, or having at least one and two starvations vs the file size N \label{noonetwo}}
\end{minipage}
\end{figure}
The mathematical models of QoE metrics are computed and compared with event-driven simulations using MATLAB. A timer generates random variable stamps that record the frames arrival and departure. Each frame contains two blocks (original frame and the copy of another frame). We monitor the playout buffer length based on the frames arrival rate and the playback rate. A switching occurs when there is no more frames of high bitrate. In this case, the player display only the frames with the low bitrate until high bitrate frames are available. A starvation happens when the buffer is empty. We take into account the size of the video file. We vary the parameters frame arrival rate $\lambda$, service rate $\mu$, the file size $N$, the pre-fetching threshold $x$, the offset $\phi$. We run each set of simulations for 4000 times and we show here only a small set of figures obtained, due to page number constraints.
Our model exhibits excellent accuracy with the simulations. Figure \ref{starvfile} shows the probability of starvation given the file size N, for different settings of the traffic intensity $\rho$. The start-up threshold $x=40$ while the offset $\phi=50$. The probability of starvation decreases when the network throughput increases.  The streaming users only suffer from the playback interruption when $\rho < 1$. 
We further evaluate the probabilities of having no starvation, at least one and two starvations, given the file size $N$ for $\rho=0.66$. Figure \ref{noonetwo} shows that our analytical model predicts the starvation probabilities accurately. The probability of no starvation decreases from 1 to 0, while the probability of having at least one and two starvations increases. The lag between the two curves gives an idea about the mean playback time, i.e., the mean time between two consecutive starvations. Then, based on the file size, the initial buffering latency, the offset and the traffic intensity, the model predicts the number of starvation that the video session could have.
\begin{figure}[t]
\begin{center}
\includegraphics[scale=0.4]{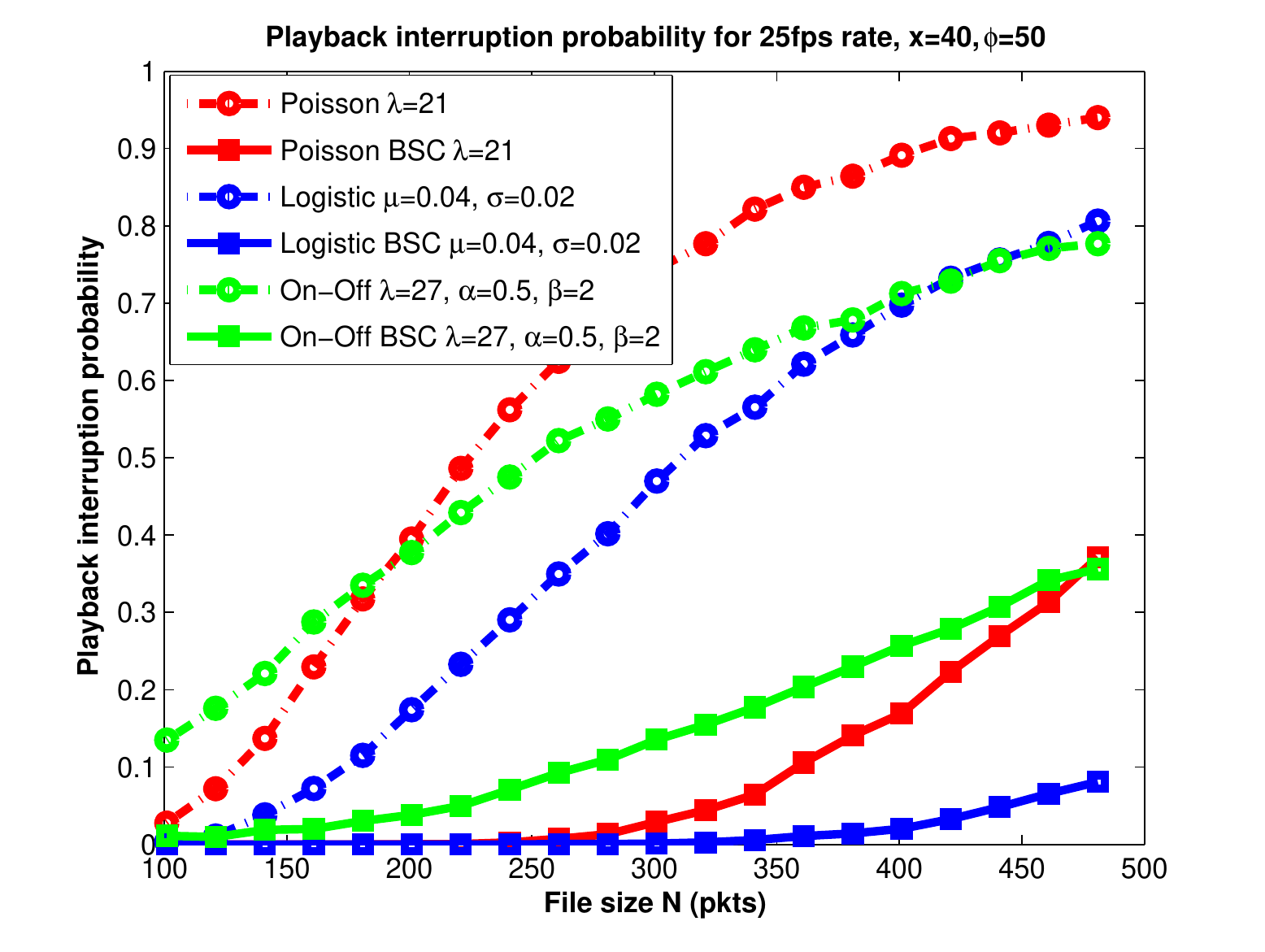}
\caption{The playback interruption probability using several packet arrivals processes}
\label{dist}
\end{center}
\end{figure}

Before studying the QoE optimization problem, we evaluate the BSC scheme  under different types of arrival processes:  the logistic  distribution and the  ON/OFF arrival process.   In Fig. \ref{dist} we show that the  BSC scheme reduces the  playback interruption. Furthermore, we obtain an important improvement of  the probability of playback interruption where the arrival of frames follows the logistic process. 

\subsection{QoE Optimization}
\label{QoEOPT}
We consider five video resolutions (1080p, 720p, 480p, 360p and 240p) with the corresponding bitrates (4500Kbps, 2500Kbps, 1000Kbps, 750Kbps and 400Kbps). The network throughput is 2200Kbps or 3000Kbps. Then we compute the corresponding traffic intensities $\rho$. We also compute the bitrate in the BSC scheme using H.264/SVC codec. We compare the QoE metrics between a classical DASH based SVC and the BSC system. Our BSC system is based on the SVC codec, then, we compare it to DASH/SVC since we know that SVC adds 10\% encoding overhead compared with the same quality AVC. This comparison is done on a single video segment. The bitrate adaptation in the BSC system is the scope of our ongoing research. We will show the two appropriate bitrates to select after the downloading of each segment, and compare the video quality to the classical DASH system. \\
Let assume that the network throughput is 2200Kbps. In DASH, the adaptation engine will select the bitrate that is just under the network throughput, i.e., 480p resolution. What happens if we select 480p$+$360p, 720p$+$360p or 720p$+$480p in the BSC scheme? \\
When we select 480p$+$360p in the BSC, the cost of DASH is better than BSC although there is no starvation in both cases. Indeed, the DASH system just benefits from the initial buffering latency. However, selecting the resolutions 720p$+$360p or 720p$+$480p allows to minimize the QoE cost function (Fig. \ref{720p}). There is only one starvation for a file size of 1500 frames but with a better rendering quality. Hence, when the network throughput changes to 3000Kbps, the BSC can use 1080p video resolution while the simple DASH cannot.
\begin{figure}[t]
\begin{minipage}[h]{0.48\linewidth}
\centering\epsfig{figure=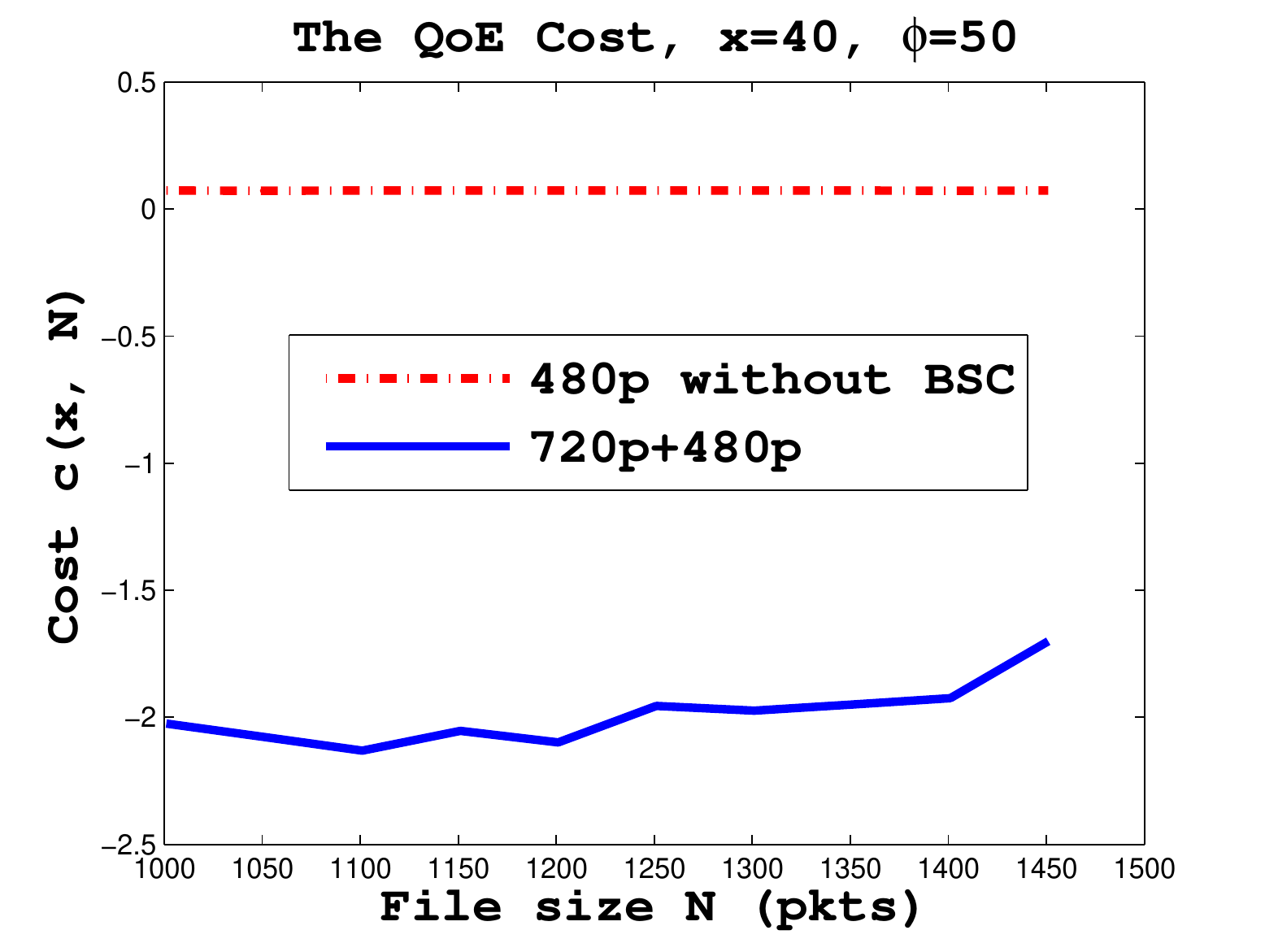, width=\linewidth}
\caption{The QoE cost function for 480p and 720p$+$480p, $\gamma_1=0.1$, $\gamma_2=1$, $\gamma_3=0.01$ \label{720p}}
\end{minipage}
\begin{minipage}[h]{0.48\linewidth}
\centering\epsfig{figure=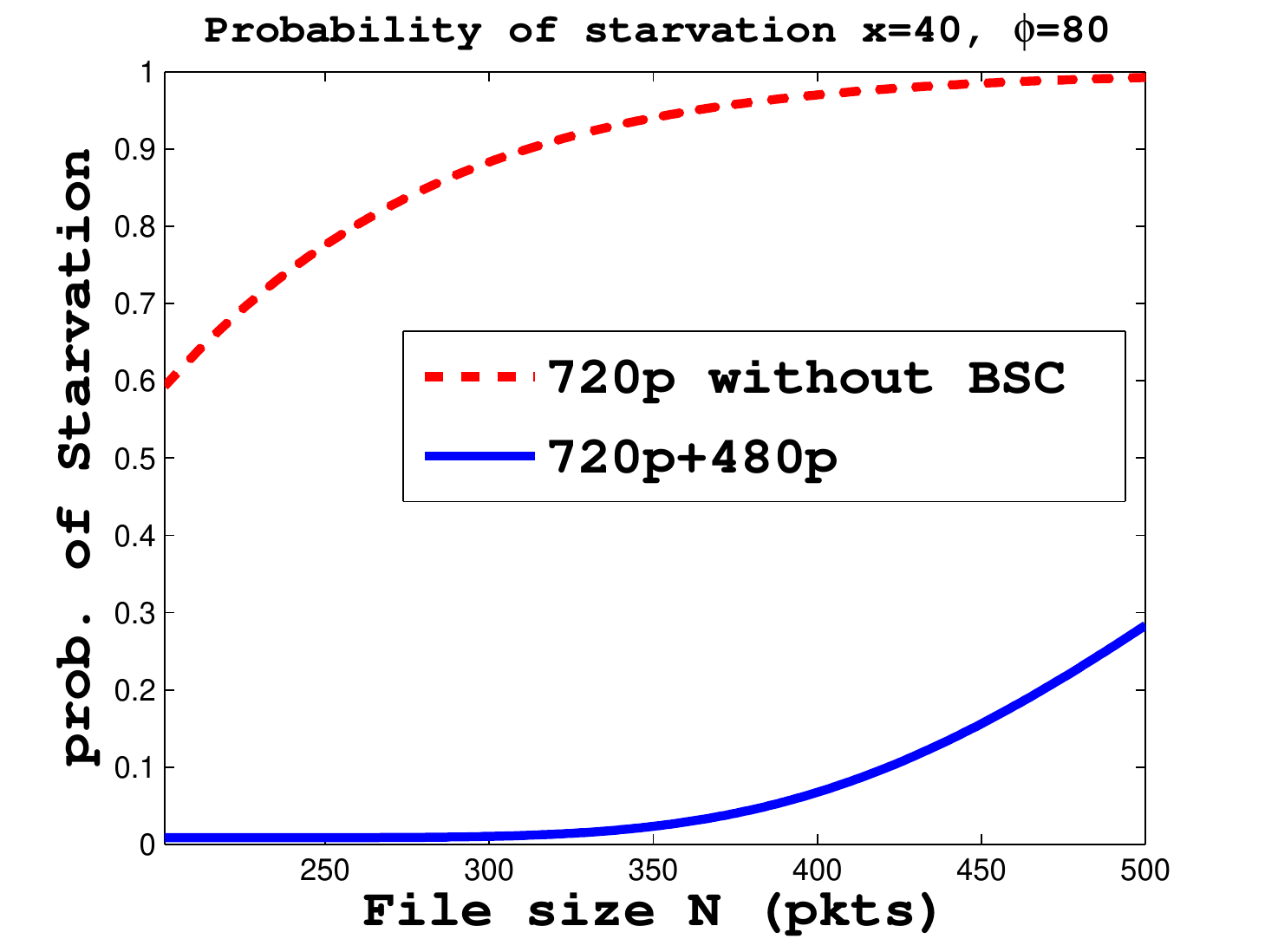, width=\linewidth}
\caption{Reduce the probability of starvation with BSC scheme \label{low_high}}
\end{minipage}
\end{figure}

\section{Conclusion and Discussion}
\label{conclusion}
In this paper, we proposed the novel Backward-Shifted Coding (BSC) scheme to improve the performance of HTTP adaptive streaming. Inspired from Forward Error Correction, BSC adds time-shifted redundancy to the video stream that provides smooth playback even if the main stream is interrupted.

We describe an integration of BSC into SVC and its operation with HAS. We then provide an analytical characterization of the dominating QoE factors initial buffering time, starvation and average video bitrate. We compute the first two metrics using the Ballot theorem and obtain explicit results for the probability generation function of the starvation. The average video bitrate metric is computed using the quasi-stationary approach. Finally, we show by global optimization that BSC can improve both the playback interruption and the average bitrate in video streaming systems compared to standard DASH on a single video segment. 

We will explore the bitrate adaptation in the BSC system in our future work.

\section*{Acknowledgment}
\normalsize{This work has been carried out in the framework of IDEFIX project, funded by the ANR under the contract number ANR-13-INFR-0006.}

\end{document}